\documentclass[conference,letterpaper]{IEEEtran}
\usepackage[english]{babel}

\usepackage{cite}
\usepackage{amsmath,amssymb,amsfonts,amsthm}
\theoremstyle{definition}

\newtheorem{prop}{Proposition}

\usepackage{algorithmic}
\usepackage{graphicx}
\usepackage{textcomp}
\usepackage{xcolor}
\usepackage{float}
\usepackage{subcaption}

\usepackage{array}
\usepackage[caption=false,font=footnotesize]{subfig}
\usepackage{lipsum}

\usepackage[linesnumbered,boxed]{algorithm2e} 
\usepackage{mathrsfs}
\usepackage{color}%

\usepackage{scrextend}

\usepackage{soul}
\usepackage{multicol}

\newcommand{\compress}{\vspace{-2pt}}
\newcommand{\compressfloat}{\vspace{-2pt}}

\def\BibTeX{{\rm B\kern-.05em{\sc i\kern-.025em b}\kern-.08em
    T\kern-.1667em\lower.7ex\hbox{E}\kern-.125emX}}

    \RestyleAlgo{ruled}
\SetKwComment{Comment}{/* }{ */}

\begin{document}

\title{ Subgraph Matching via Partial Optimal Transport
\thanks{This work was supported by the Knut and Alice Wallenberg
foundation under grant KAW 2021.0274.}
}

\author{\IEEEauthorblockN{Wen-Xin Pan}
\IEEEauthorblockA{\textit{dept. name of organization (of Aff.)} \\
\textit{name of organization (of Aff.)}\\
City, Country \\
email address or ORCID}
\and
\IEEEauthorblockN{Isabel Haasler}
\IEEEauthorblockA{\textit{Signal Processing Laboratory 4} \\
\textit{EPFL}\\
Lausanne, Switzerland \\
isabel.haasler@epfl.ch}
\and
\IEEEauthorblockN{Pascal Frossard}
\IEEEauthorblockA{\textit{Signal Processing Laboratory 4} \\
\textit{EPFL}\\
Lausanne, Switzerland \\
pascal.frossard@epfl.ch}

\author{
Wen-Xin Pan, Isabel Haasler, Pascal Frossard \\
Signal Processing Laboratory (LTS4), EPFL\\
Emails:\{wenxin.pan, isabel.haasler, pascal.frossard\}@epfl.ch
}

}

\newpage

\maketitle

\begin{abstract}
In this work, we propose a novel approach for subgraph matching, the problem of finding a given query graph in a large source graph, based on the fused Gromov-Wasserstein distance.
We formulate the subgraph matching problem as a partial fused Gromov-Wasserstein problem, which allows us to build on existing theory and computational methods in order to solve this challenging problem.
We extend our method by employing a subgraph sliding approach, which makes it efficient even for large graphs.
In numerical experiments, we showcase that our new algorithms have the ability to outperform state-of-the-art methods for subgraph matching on synthetic as well as real-world datasets.
In particular, our methods exhibit robustness with respect to noise in the datasets and achieve very fast query times. \looseness=-1
\end{abstract}

\section{Introduction}

 Subgraph matching is a common problem in information sciences, where one aims to retrieve a predefined query graph within a large source graph.
This is also an important problem in various applications, such as
 anomaly detection \cite{akoglu2015graph,bakirtas_database_2023} and knowledge discovery \cite{anchuri2013approximate}.
Herein, the graphs may represent for instance knowledge graphs \cite{sun2022subgraph}, shapes \cite{sun_survey_2020}, biological data \cite{tian_saga_2007}, or social networks \cite{dasgupta_discovering_2020}.
Subgraph matching has traditionally been addressed by modifications of the quadratic assignment problem \cite{schellewald_probabilistic_2005,lawler_quadratic_1963}. 
This requires solving a combinatorial optimization problem, and is computationally prohibitive for large graphs.
To tackle these computational issues, various index-based methods have been developed \cite{tian_saga_2007,khan_nema_2013,liu_g-finder_2019},
but these typically rely on heuristics.
Both types of approaches for subgraph matching face a particular challenge in
the setting where no exact matching between the query graph with any subgraph in the source graph is possible.
For instance, this may be the case when the source graph stems from measurements that may be disturbed by noise, which is a prevalent issue in many biological applications \cite{tian_saga_2007}.
In this noisy setting, it is crucial to describe the deviation between the query graph and the matched subgraph appropriately in order to retrieve the most relevant subgraph within the source graph.

In this work we propose a novel approach for subgraph matching based on the fused Gromov-Wasserstein distance, which is a recently introduced distance for graphs \cite{titouan2019optimal, vayer2020fused}, with strong mathematical and computational foundations.
In particular, the fused Gromov-Wasserstein distance is based on the Wasserstein metric \cite{villani2021topics, peyre2019computational} and can be used to define a metric on the space of structured objects. Moreover, it can be seen as a relaxation and generalization of the quadratic assignment problem \cite{titouan_optimal_2019}.
Within this framework, we introduce two novel methods for subgraph matching.
First, Subgraph Optimal Transport (SOT) poses the subgraph matching problem as a partial fused Gromov-Wasserstein problem, where we find the subgraph within the source graph that is most similar to the query graph, as measured by the fused Gromov-Wasserstein distance.
We provide an efficient algorithm for solving this problem.
Second, we present a refinement of this method, called Sliding Subgraph Optimal Transport (SSOT), inspired by ideas from index-based methods, which is computationally efficient even for large graphs.
Numerical experiments on synthetic and real-world datasets show that our methods can outperform state-of-the-art methods for subgraph matching in terms of success rate and query time, and are in particular robust to feature noise.

The paper is structured as follows.
In Section~\ref{sec:background}, we introduce relevant background material on partial optimal transport and the fused Gromov-Wasserstein problem.
In Sections~\ref{sec:sot} and \ref{sec:ssot}, we develop our two novel subgraph matching methods, SOT and SSOT.
Finally, in Section~\ref{sec:exp}, we present numerical experiments.

\compress
\section{Background} \label{sec:background}

In this section, we review the mathematical background needed for our subgraph matching framework. In particular, we present the partial optimal transport problem and the fused Gromov-Wasserstein distance.

\compress
\subsection{Optimal transport and partial optimal transport} \label{subsec:ot}

Optimal transport is a classical problem in mathematics \cite{villani2021topics}, which is often used as a measure of distance between probability distributions \cite{peyre2019computational}.
In this work, we consider the discrete version of optimal transport.
In more details, let $ \boldsymbol{p} \in \mathbb{R}^n$ and $ \boldsymbol{q} \in \mathbb{R}^m$ be two probability vectors with support on the set of points $x^{(1)},\dots,x^{(n)} $ and $y^{(1)},\dots,y^{(m)}$, respectively, which lie in the same space.
One can define a cost matrix $\boldsymbol{M} \in \mathbb{R}^{n \times m}$, where the element $\boldsymbol{M}_{ij}$ measures the distance between $x^{(i)}$ and $y^{(j)}$.
The optimal transport problem aims to find an optimal transport plan, which is a non-negative matrix of the same dimensions, $\boldsymbol{T} \in \mathbb{R}_+^{n \times m}$, where $\boldsymbol{T}_{ij}$ denotes the amount of mass transported from $x^{(i)}$ to $y^{(j)}$.
The total cost associated with a given transport plan $\boldsymbol{T}$ is
\begin{equation*}
    \langle \boldsymbol{M}, \boldsymbol{T} \rangle := 
    \text{trace} \left( \boldsymbol{M}^\top \boldsymbol{T} \right) = \sum_{\substack{i=1,\dots,n \\ j=1,\dots,m}} \boldsymbol{M}_{ij} \boldsymbol{T}_{ij}.
\end{equation*}
A transport plan $\boldsymbol{T}$ defines a feasible transport between $\boldsymbol{p}$ and $\boldsymbol{q}$ if it lies in the set
\begin{equation*}
    \mathcal{T}(\boldsymbol{p},\boldsymbol{q}) := \left\{ 
	\boldsymbol{T}\in \mathbb{R}_{+}^{n \times m} \ | \ 
	\boldsymbol{T}  \boldsymbol{1}_m={\boldsymbol{p}},\  \boldsymbol{T} ^{\top} \boldsymbol{1}_n={\boldsymbol{q}}
	\right\},
\end{equation*}
where $\boldsymbol{1}_d \in \mathbb{R}^d$ denotes a vector of ones\footnote{To simplify notation we sometimes omit the subindex when the size of the vector is clear from the context.}.
The optimal transport problem is to find the most cost-efficient transport plan, that is, to solve the linear program
\begin{equation} \label{eq:OT}
    \min_{\boldsymbol{T} \in \mathcal{T}(\boldsymbol{p},\boldsymbol{q})} \langle \boldsymbol{M}, \boldsymbol{T} \rangle.
\end{equation}
Note that if the points $x^{(1)},\dots,x^{(n)} $ and $y^{(1)},\dots,y^{(m)} $ lie in a metric space, and $\boldsymbol{M}$ describes the metric distances between them, then the objective value of \eqref{eq:OT} defines a metric on the space of discrete measures on this space, called the Wasserstein distance \cite{villani2021topics, peyre2019computational}.

Partial optimal transport \cite{caffarelli2010free, figalli2010optimal, chapel2020partial} is an extension of the classical optimal transport problem, where the feasibility set in \eqref{eq:OT} is replaced by
\begin{equation*}
    \mathcal{T}_{s}(\boldsymbol{p}, {\boldsymbol{q}}) :=\!  \left\{{\boldsymbol{T}} \in \mathbb{R}_{+}^{n \times m}  \mid {\boldsymbol{T}} \mathbf{1} \leq \boldsymbol{p},\ {\boldsymbol{T}}^{\top} \!\mathbf{1} \leq {\boldsymbol{q}},\ \mathbf{1}^{\top} {\boldsymbol{T}} \mathbf{1} =s\right\} \!.
\end{equation*}
The parameter $s\in (0,1]$ determines the total amount of mass that must be moved.
By utilizing this feasibility set, we can consider the case where not all of the mass must be transported, and where the distributions $\boldsymbol{p}$ and $\boldsymbol{q}$ can possibly have unequal mass, 
as it may happen, e.g., in computer vision settings \cite{rubner2000earth, pele2009fast, sarlin2020superglue}.

\compress
\subsection{Fused Gromov-Wasserstein distance for graphs}
The Gromov-Wasserstein distance allows for comparing distributions within two different metric spaces, by comparing the correspondences of data points within each space \cite{memoli2011gromov, peyre2016gromov}.
This can be extended to graphs, where
the Gromov-Wasserstein problem %for graphs 
analogously considers the difference between each pair of nodes within each graph to measure 
the structural difference between graphs
\cite{xu2019gromov}. \looseness =-1

More formally, a graph is a tuple $G=(V,E)$, which consists of a set of nodes $V$, which are connected by edges in the set $E$.
Let $G_s$ and $G_t$ be two graphs with $n$ and $m$ nodes, respectively.
Analogously to the optimal transport problem \eqref{eq:OT}, we seek a transport plan $\boldsymbol{T}\in\mathcal{T}(\boldsymbol{p},\boldsymbol{q})$, where $\boldsymbol{p}$ and $\boldsymbol{q}$ describe the importance of the nodes in each graph.
This transport plan provides a "soft" matching between the two graphs in the sense that it matches mass from the nodes of the source graph to the nodes of the target graph.
In particular, if $\boldsymbol{T}$ is a permutation matrix, then the transport plan describes a one-to-one matching between the nodes of the two graphs.
The structure of the graphs is represented by matrices $\boldsymbol{C}^s \in \mathbb{R}_+^{n\times n}$ and $\boldsymbol{C}^t \in \mathbb{R}_+^{m\times m}$, which can for example describe the shortest path distance between each pair of nodes, or represent the adjacency matrix.
Given two structure matrices $\boldsymbol{C}^s \in \mathbb{R}^{n\times n}$ and $\boldsymbol{C}^t \in \mathbb{R}^{m\times m}$, one can define a tensor $\boldsymbol{L}^{G_s,G_t}\in\mathbb{R}^{n \times n\times m\times m}$ with entries
\begin{equation} \label{eq:L_cost}
    \boldsymbol{L}^{G_s,G_t}_{i,i',j,j'}= \left( \boldsymbol{C}^s_{i,i'} - \boldsymbol{C}^t_{j,j'} \right)^2,
\end{equation}
that represents the %difference
structural discrepancy between nodes pairs $(i,i')$ in $G_s$ and $(j,j')$ in $G_t$. 
The total cost of a transportation plan $\boldsymbol{T}$ is measured as
\begin{equation*}
    \sum_{i,i',j,j'} \boldsymbol{L}^{G_s,G_t}_{i,i',j,j'} \boldsymbol{T}_{ij} \boldsymbol{T}_{i',j'} = \langle \boldsymbol{L}^{G_s,G_t} \otimes \boldsymbol{T}, \boldsymbol{T} \rangle,
\end{equation*}
where $\otimes$ is defined as 
\begin{equation} \label{eq:tensormatrix}
    \left( \boldsymbol{L} \otimes \boldsymbol{T}  \right)_{i,j} := \sum_{i',j'} \boldsymbol{L}_{i,i',j,j'} \boldsymbol{T}_{i',j'}.
\end{equation}
The Gromov-Wasserstein distance between the graphs $G_s$ and $G_t$ is defined by the optimal value of the optimization problem
 \begin{equation*} 
 \min_{\boldsymbol{T}\in\mathcal{T}(\boldsymbol{p},\boldsymbol{q})}\langle \boldsymbol{L}^{G_s,G_t} \otimes \boldsymbol{T}, \boldsymbol{T} \rangle.
\end{equation*} 

This distance between graphs has been generalized to also take into account graphs that may have node features in the fused Gromov-Wasserstein problem \cite{titouan2019optimal, vayer2020fused}. 
With node features that lie in a metric space, one can define a cost matrix $\boldsymbol{M}\in\mathbb{R}^{n \times m}$, similarly to Section~\ref{subsec:ot}, i.e., $\boldsymbol{M}_{ij}$ describes the distance between the features of node $i$ in $G_s$ and node $j$ in $G_t$.
Then the optimal transport problem between two labeled graphs $G_s$ and $G_t$ is defined as
\begin{equation} \label{eq:FGW}
    \min_{\boldsymbol{T}\in\mathcal{T}(\boldsymbol{p},\boldsymbol{q})} \ (1-\alpha)\langle\boldsymbol{T}, \boldsymbol{M}\rangle +\alpha \langle \boldsymbol{L}^{G_s,G_t} \otimes \boldsymbol{T}, \boldsymbol{T} \rangle .
\end{equation}
The optimal transport plan solving \eqref{eq:FGW} describes a soft correspondence between the nodes of two graphs, such that both their features and structure are similar.
Under suitable conditions it satisfies metric properties over a space of structured data
\cite{titouan2019optimal},
and thus provides a mathematically rigorous framework for 
measuring distances between graphs.

\compress
\section{SOT: Subgraph optimal transport} \label{sec:sot}

We now present our first solution to the subgraph matching problem, which is based on a partial fused Gromov-Wasserstein problem.
Our second solution is a refinement of this method and is described in the next section.

Let us consider a source graph $G_s$ with $n$ nodes, in which we want to find a subgraph that is similar to a given query graph $G_q$ with $m<n$ nodes. 
We propose to measure this similarity by the fused Gromov-Wasserstein distance. 
More precisely, we want to identify the subgraph of $m$ nodes in $G_s$ that is closest to the query graph $G_q$ as measured by the fused Gromov-Wasserstein distance.
Following the idea of partial optimal transport, we thus want to find a transport plan $\boldsymbol{T}$ between $m$ nodes in the source graph and all nodes of the query graph.
We give equal importance to all nodes, and thus define the mass of each node in the source and in the query graph as $1/n$. That is, $\boldsymbol{p} = \frac{1}{n} \boldsymbol{1}_n$ and $\boldsymbol{q} = \frac{1}{n} \boldsymbol{1}_m$. In particular, note that the target distribution has total mass $\boldsymbol{q}^\top \boldsymbol{1}_m = m/n$. In order to make sure that the transport plan sends mass $1/n$ to every target node we require that $\boldsymbol{T} \in \mathcal{T}_{s}(\boldsymbol{p}, {\boldsymbol{q}})$, where $s=m/n$. Note that this constraint ensures that $\boldsymbol{T}^\top \boldsymbol{1}=\boldsymbol{q}$.
Thus, we define the subgraph matching problem as finding the partial fused Gromov-Wasserstein distance between $G_s$ and $G_q$, defined as follows
\begin{equation} \label{eq:partialFGW}
    \min_{\boldsymbol{T}\in\mathcal{T}_{m/n}(\boldsymbol{p},\boldsymbol{q})} \ (1-\alpha)\langle\boldsymbol{T}, \boldsymbol{M}\rangle +\alpha \langle \boldsymbol{L}^{G_s,G_q} \otimes \boldsymbol{T}, \boldsymbol{T} \rangle .
\end{equation}

This problem can also be formulated as a standard fused Gromov-Wasserstein problem \eqref{eq:FGW} by adding a dummy node of mass $1-m/n$ to the query graph, as illustrated in Figure~\ref{fig:sot}.
\begin{figure}[tb]
\centering\includegraphics[width=1\linewidth]{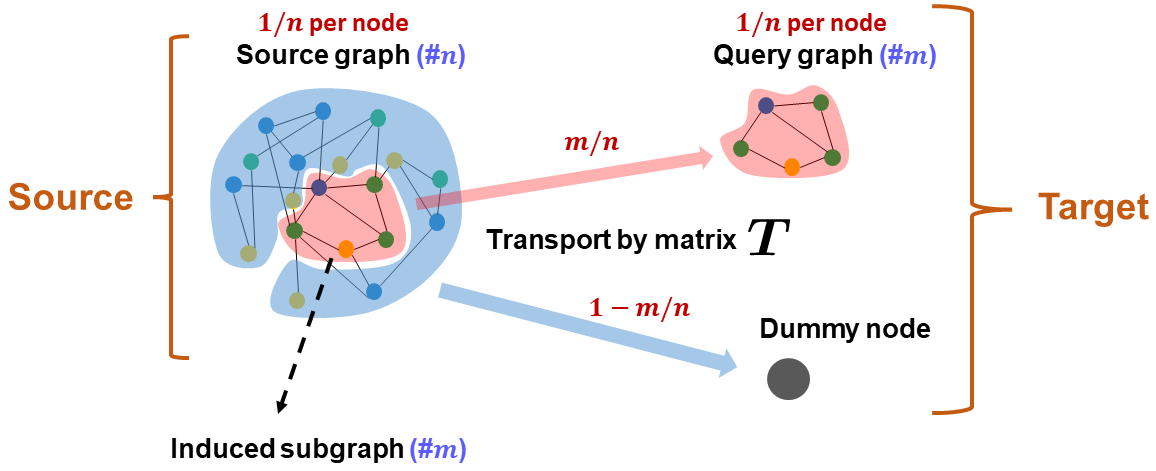}
\caption{Illustration of Subgraph Optimal Transport (SOT).
} %\vspace{-4pt}
\label{fig:sot}
\compressfloat
\end{figure}
More precisely, define the target graph $G_t$ as a graph with $m+1$ nodes, which consists of the query graph $G_q$ augmented by an unconnected dummy node.
Without loss of generality, we let the dummy node be the final node in our ordering. Then, we define a target probability distribution $\boldsymbol{\hat q} = [ \boldsymbol{ q}^\top, 1 - m/n]^\top \in \mathbb{R}^{m+1}$. 
Moreover, we allow for mass to be transported to the dummy node "for free", that is we define the feature cost matrix $\boldsymbol{\hat M} \in \mathbb{R}^{n \times (m+1)}$ and structure cost tensor $\boldsymbol{\hat L}^{G_s,G_t} \in\mathbb{R}^{n \times n\times (m+1)\times (m+1)}$ as $\boldsymbol{\hat M} = [\boldsymbol{M}, \boldsymbol{0}_n]$, where $\boldsymbol{M}$ is the feature cost matrix in \eqref{eq:partialFGW}, and
\begin{equation} \label{eq:Lhat}
	\boldsymbol{\hat{L}}^{G_s,G_t}_{i,i',j,j'}=
	\begin{cases}
	   \boldsymbol{L}^{G_s,G_q}_{i,i',j,j'}, & \text{if } j,j' \in \{1,...,m\} \\
		0, & \text{if } j= m+1 \text{ or }  j'= m+1,
	\end{cases}
\end{equation}
where $\boldsymbol{L}^{G_s,G_q}$ is the structure cost tensor in \eqref{eq:partialFGW}.

We note that introducing dummy nodes has previously been proposed in other graph matching problems \cite{medasani_graph_2001,gold_graduated_1996}.
Moreover, they have been used for the partial optimal transport problem and for the partial Gromov-Wasserstein problem \cite{chapel2020partial}.
However, in contrast to \cite{chapel2020partial} the transport costs to the dummy node are defined to be zero for our subgraph matching application, which leads to a more efficient computational method.
Namely, we can rely on the methods developed in previous works \cite{peyre_gromov-wasserstein_2016, titouan_optimal_2019,flamary2021pot}, which solve \eqref{eq:FGW} by means of a Frank-Wolfe algorithm.
We describe the method and our extensions to it in the Appendix~\ref{sec:alg}.

\compress
\section{SSOT: Sliding Subgraph Optimal Transport} \label{sec:ssot}

We note that the fused Gromov-Wasserstein problem \eqref{eq:FGW} that we have to solve for the SOT framework is a non-convex optimization problem with $n^2(m+1)^2$ variables, where $n$ and $m$ are the numbers of nodes in the source graph and query graph, respectively.
For large source graphs, this problem thus becomes computationally infeasible.
Moreover, the energy landscape 
may become increasingly complex, and numerical solvers might easily get stuck in local minima. 
To address these challenges we propose a refinement of the SOT framework, namely the Sliding Subgraph Optimal Transport (SSOT) method.
Note that in SOT the optimal transport plan maps all nodes in the source graph, which are not in the subgraph that is most similar to the query graph, to the dummy node at zero cost.
Thus, only $m$ nodes contribute to the optimal value in \eqref{eq:FGW} (and \eqref{eq:partialFGW}).
Therefore, we propose to iteratively compare the query graph with small subgraphs in the source graph, as illustrated in Figure~\ref{fig:ssot}.
\begin{figure}[tb]
\centering\includegraphics[width=1\linewidth]{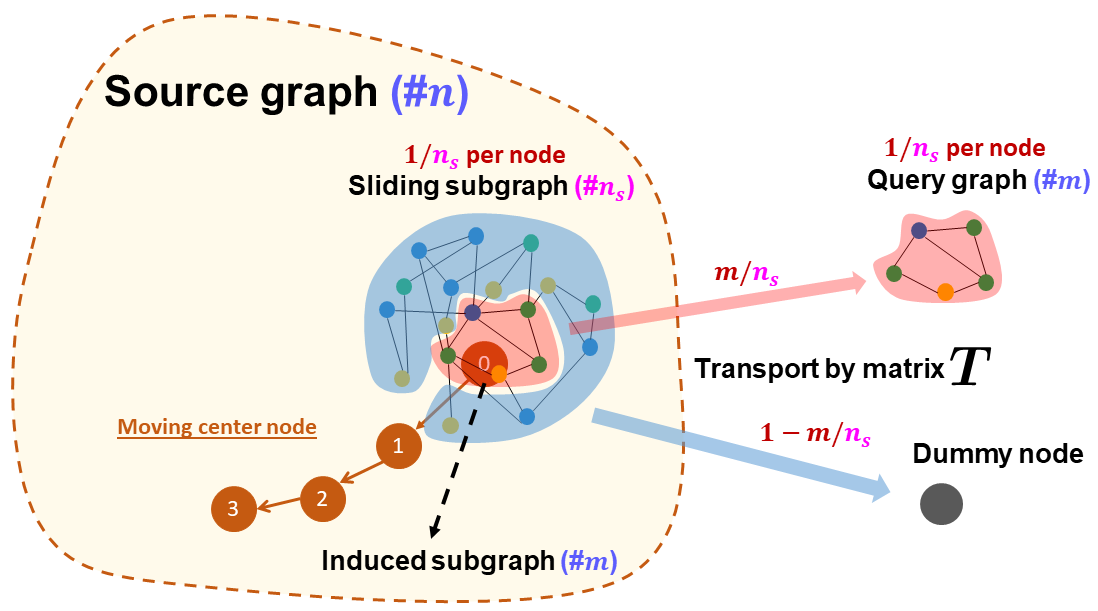}
\caption{Illustration of Sliding Subgraph Optimal Transport (SSOT). 
}
\label{fig:ssot}
\compressfloat
\end{figure}
More precisely, we iterate over the nodes in the source graph, and solve the partial fused Gromov-Wasserstein problem between a \emph{sliding subgraph}, defined by the neighborhood of this node, and the query graph.  \looseness=-1

\begin{algorithm}[tb]
	\caption{SSOT}\label{alg:ssot}
	\KwIn{$ {G}_s$, ${G}_q $}
    Compute $k$ by \eqref{r}\;
	\For{{\rm node} $v$ {\rm in} $ G_s $}{
		Define $ {G}_{v} $ as the $k$-hop neighborhood of $ v$\;
		\If{ \eqref{eq:criteria} is true }{
                Compute partial fused Gromov-Wasserstein distance \eqref{eq:partialFGW} between $G_{v}$ and $G_q$
		}		
    	}
 $G^* \leftarrow$ subgraph corresponding to the smallest computed partial fused Gromov-Wasserstein distance\;
\KwOut{$G^*$}
\end{algorithm}
\begin{figure*}
    \centering
    \begin{subfigure}[b]{0.32\textwidth}
		\centering
		 \includegraphics[width=\textwidth]{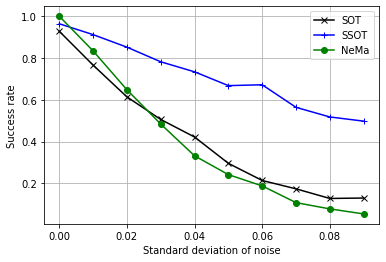}
		\caption{Success rate with feature noise.}
		\label{3}
	\end{subfigure}
    \begin{subfigure}[b]{0.32\textwidth}
		\centering
		\includegraphics[width=\textwidth]{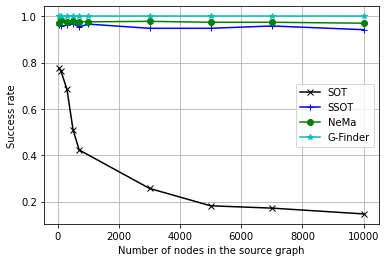}
		\caption{Success rate of exact matching.}
		\label{1}
	\end{subfigure}
	\begin{subfigure}[b]{0.32\textwidth}
		\centering
		\includegraphics[width=\textwidth]{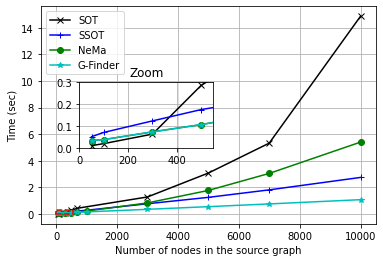}
		\caption{Average query time.}
		\label{2}
	\end{subfigure}
 \caption{Performance of our algorithms on Erdős–Rényi graphs.} %\vspace{-4pt}
    \label{fig:results}
    \compressfloat
\end{figure*}
The SSOT method is summarized in Algorithm~\ref{alg:ssot}.
For each node $v_s$ in the source graph, we construct a sliding subgraph $G_{v_s}$ as its $k$-hop neighborhood, where
\begin{equation}\label{r}
k = \min_{v_i \in V_q} \{ \max_{v_j \in V_q} d(v_i,v_j) \},
\end{equation}
is defined as the radius of the query graph, and where $V_q$ denotes the node set of the query graph, and $d(\cdot,\cdot)$ denotes the shortest path distance on $G_q$.
To further improve the computational efficiency of our method, we propose to prune out unsuitable
candidates among the sliding subgraphs 
before we compute the partial fused Gromov-Wasserstein distance.
A "good" candidate $G_{v_s}$ should have at least as many nodes and similar node features as the query graph $G_q$, and thus satisfy the following two criteria:
\begin{equation} \label{eq:criteria}
    \begin{aligned}
        & \text{1) } G_{v_s} \text{ has at least as many nodes as } G_q, \\
        & \text{2) } \min_{\boldsymbol{T} \in \mathcal{T}_{m/n_s}(\boldsymbol{p},\boldsymbol{q})} \langle \boldsymbol{\hat M}, \boldsymbol{T} \rangle < T^W,
    \end{aligned}
\end{equation}
where $T^W$ is a predefined positive threshold and $\boldsymbol{\hat M}$ is the feature cost matrix as defined in Section~\ref{sec:sot}.
Note that the second criterion computes the partial optimal transport only between the features of the subgraph and the query graph, which requires solving a small linear program and is very cheap compared to solving the partial fused Gromov-Wasserstein problem of \eqref{eq:partialFGW}. 

To summarize, SSOT improves the computational burden of SOT by solving partial fused Gromov-Wasserstein problems \eqref{eq:partialFGW} of smaller size.
By construction, we have that the average number of nodes in the sliding subgraph is $n_s \sim \mathcal{O}( {\text{deg}(G_s)}^k )$, where $\text{deg}(G_s)$ denotes the average node degree in the source graph.
Thus, the optimization problems involved in SSOT are on average of the order $\mathcal{O}(n_s^2 (m+1)^2)$.
SSOT requires solving at most $n$ optimization problems of this size, %
and the number of problems that need to be solved is small if many sliding subgraphs do not satisfy the filtering criteria \eqref{eq:criteria}.
These two observations show that SSOT is especially efficient for sparse graphs, and for graphs with very expressive features, where we can choose $T^W$ small enough such that many non-optimal subgraphs do not satisfy criterion 2) in \eqref{eq:criteria}. 
Moreover, we note that SSOT may find a smaller minimum than SOT in cases where SOT gets stuck in local minima.

\compress
\section{Experiments} \label{sec:exp}

We showcase the performance of our suggested methods for subgraph matching on Erdős–Rényi graphs and several real-world datasets in the setting of exact matching, as well as when there is feature noise.
We compare our results with the two state-of-the-art index-based methods NeMa \cite{khan_nema_2013} and G-Finder \cite{liu_g-finder_2019}. 
G-Finder finds a subgraph that matches the features of the query graph exactly, whereas NeMa allows for deviations in the node features.
For SOT and SSOT, the structure cost tensor $\boldsymbol{L}^{G_s,G_t}$ is defined using the graphs' adjacency matrices as structure matrices $\boldsymbol{C}^s$ and $\boldsymbol{C}^t$. 
The elements $\boldsymbol{M}_{ij}$ of the feature cost matrix and the feature similarity in NeMa are defined as follows:
For real-valued vector features $\boldsymbol{x_i}$, $\boldsymbol{x_j}$, we use the normalized square $L^2$-norm distance $1- ( 1+\left\|\boldsymbol{x_i}-\boldsymbol{x_j}\right\|_2^2 )^{-1}$, and for integer set features $A_i$, $A_j$, we use the Jaccard dissimilarity $1- |A_i \cap A_j| /|A_i \cup A_j|$.
In order to emphasize features and structure equally, we normalize both terms in \eqref{eq:partialFGW} as described in the Appendix~\ref{app:nFGW}, and set the trade-off parameter to $\alpha=0.5$ in all experiments.
The experiments are conducted on AMD Ryzen 7 PRO 4750U CPU with 16GB RAM. 

\compress
\subsection{Experiments on Erdős–Rényi graphs} \label{subsec:expsyn}

We perform subgraph matching within Erdős–Rényi graphs that are constructed as follows.
As a query graph, we first construct an Erdős–Rényi graph $G_q$ of size $m=5$ with edge probability $0.5$. The average degree within the query graph is thus $2$.
To construct a source graph with $n=100$ nodes we augment this query graph with $95$ nodes.
The nodes are connected with edge probability $3/(n-1)$
in order to ensure that the source graph has average node degree $d_s=3$. We do not add any edges between the nodes of the original query graph.
We uniformly sample features for each node in the source graph from $[0,1]$, and copy the corresponding features to the nodes in the query graph.
We then add Gaussian noise to the features of the query node with varying standard deviations. 
We solve the subgraph matching problem with our methods and NeMa.
For SSOT we set the threshold $T^W=1$, and use the same value for the feature threshold in NeMa, in order to allow for noise in the features.
For each noise level, the experiment is repeated 500 times,
and the success rates for our methods and for NeMa are shown in Figure~\ref{3}.
The success rate denotes the ratio of trials, where the query graph is identified within the source graph. 
We see that in the noise-free case NeMa performs slightly better than our methods. 
However, as the noise increases, our methods quickly outperform NeMa.
In particular, SSOT achieves significantly higher success rates than NeMa in the case of high noise levels.

Next, we study the performance of our methods with respect to the size of the source graph in a setting without feature noise.
The Erdős–Rényi query graphs and source graphs are constructed as before, and we test different sizes of source graphs.
Moreover, we now assign less expressive node features, which are picked from the set $\{ \frac{k}{20} | k=1,\dots,20 \}$.
In this noise-free setting, we can utilize the fact that node features can be matched exactly and thus set the threshold for SSOT to $T^W=10^{-9}$, and use the same value for the feature threshold in NeMa.
We run 500 trials of this experiment and compare the success rate and average query time of our tested methods.
The results for varying numbers of nodes in the source graph are summarized in Figures~\ref{1} and \ref{2}. 
One can see that the success rate of SOT decreases quickly as the graph size increases, which we assume is due to the increasingly complex energy landscape of the large fused Gromov-Wasserstein problem that needs to be solved.
Our refined SSOT method, however, is competitive with the state-of-the-art methods NeMa and G-Finder, and achieves very high success rates.
For relatively small source graphs with less than $300$ nodes, SOT runs faster than all the other methods, but for larger source graphs the size of the optimization problem leads to rapidly growing query times for SOT.
Note that the query time of SSOT grows only linearly with the number of nodes in the source graph, and only G-Finder runs faster than SSOT. However, we note that this method is based on a very efficient implementation in C++, whereas all other methods are run in Python.
Moreover, G-Finder is limited to this noise-free setting, where the exact query graph can be found in the source graph. \looseness=-1

\compress
\subsection{Experiments on real-world datasets} \label{subsec:expreal}

We now test our methods on the real-world datasets BZR \cite{KKMMN2016}, FIRSTMM\_DB \cite{KKMMN2016}, LastFM \cite{feather}, and Deezer \cite{feather}. Descriptions of these datasets can be found in Appendix~\ref{app:exp}.
For each source graph in the datasets, we randomly choose 10 subgraphs of $m=6$ nodes by the breadth-first search algorithm, which are then used as query graphs.
We test the performance of the same methods as in Section~\ref{subsec:expsyn}.
The success rates and average query times are summarized in Tables~\ref{tab:mr_exact}-\ref{tab:time_exact}.
Note that for the Deezer dataset, SOT requires too much memory to be solved on our machine.
\begin{table}[tb]
    \centering
            \caption{Success rate without noise.}
         \label{tab:mr_exact}
    \begin{tabular}{c|c|c|c|c}
     & BZR & FIRSTMM\_DB & LastFM & Deezer \\
    \hline
    SOT & 1.0 & 0.780 & 1.0 & -- \\
    SSOT & 1.0 & 0.839 & 1.0 & 1.0 \\
    NeMa & 1.0 & 0.693 & 1.0 & 0.6 \\
     G-Finder & 0.9995 & 1.0 & 1.0 & 1.0 \\
    \end{tabular}
    % \end{table}
    \vspace*{10pt}
    % \begin{table}[tb]
    %     \centering
            \caption{Average query time (in seconds) without noise.}
        \label{tab:time_exact}
     \begin{tabular}{c|c|c|c|c}
     & BZR & FIRSTMM\_DB & LastFM & Deezer \\
    \hline
    SOT & 0.005 & 4.351 & 5.659 & -- \\
    SSOT & 0.008 & 0.312 & 0.670 & 0.252 \\
    NeMa & 0.027 & 0.384 & 2.550 & 42.229 \\
     G-Finder & 0.035 & 0.388 & 3.322 & 21.669 \\
    \end{tabular}
    \compressfloat
\end{table}
All methods achieve very high success rates on BZR and LastFM.
For Deezer,
only SSOT and G-Finder have perfect success rate.
For FirstMM\_DB, G-Finder outperforms all other methods, with SSOT as second best method.
We note that all datasets have expressive node features, which explains that SOT is competitive with the other methods, in contrast to the experiments on noise-free Erdős–Rényi graphs in Section~\ref{subsec:expsyn}.
As in our previous experiments, the SOT method is the fastest for the small graph dataset BZR, but much slower than the other methods for larger graphs.
On these datasets SSOT performs impressively fast even for large graphs. In fact, for Deezer, which contains more than $28,000$ nodes, SSOT finds query graphs $100$ times faster than G-finder and $200$ times faster than NeMa. \looseness=-1

We also test our algorithms in the presence of feature noise. 
As in the previous experiments, noise is added to the features of all the nodes of the query graph. 
For real-valued features, we assign zero-mean Gaussian noise with a standard deviation $\sigma=0.5$ (for BZR), or $\sigma=0.1$ (for FIRSTMM\_DB). For integer features, the noise is assigned with $\sigma =1$ (for Deezer), or $\sigma =2$ (for LastFM), and then the feature is rounded up to an integer value.
We optimize the hyperparameter $T^W$ in SSOT and the threshold parameter in NeMa by testing values in $\{1,0.5,0.1,0.05,...,10^{-9}\}$, and report
the result that gives the highest success rate.
In case several experiments give the same success rate, we use the threshold that has the shortest average query time. 
The results are shown in Tables~\ref{tab:mr_noisy}-\ref{tab:time_noisy}, and results for other tested threshold parameters can be found in Appendix~\ref{app:exp}.
\begin{table}[tb]
         \centering
                 \caption{Success rate with noise.}
        \label{tab:mr_noisy}
    \begin{tabular}{c|c|c|c|c}
     & BZR & FIRSTMM\_DB & LastFM & Deezer \\
    \hline
    SOT & 0.264 & 0.780 & 0.9 & -- \\
    SSOT & 0.687 & 0.839 & 1.0 & 1.0 \\
    NeMa & 0.469 & 0.693 & 0.9 & -- \\
    \end{tabular}
    % \end{table}
    \vspace*{10pt}
    % \begin{table}[tb]
    %     \centering
        \caption{Average query time (in seconds) with noise.}
    \label{tab:time_noisy}
    \begin{tabular}{c|c|c|c|c}
     & BZR & FIRSTMM\_DB & LastFM & Deezer \\
    \hline
    SOT & 0.007 & 5.182 & 7.783 & -- \\
    SSOT & 0.090 & 0.322 & 99.763 & 178.976 \\
    NeMa & 0.085 & 0.389 & 358.922 & -- \\
    \end{tabular}
    \compressfloat
\end{table}
As a general trend, the presence of feature noise decreases the success rates and increases query times.
In terms of query time, SSOT achieves satisfying results, although we
note that for LastFM and Deezer the query times increase significantly compared to the noise-free setting. 
This is due to the dense graphs in these datasets,
resulting in relatively large optimization problems involved in SSOT (see Appendix~\ref{app:exp}).
In terms of success rates, SSOT outperforms all other methods, 
and is the only method that can process the largest dataset, Deezer.
As in the noise-free setting, this dataset exceeds the memory for SOT.
Moreover, for small threshold parameters NeMa fails to create node candidate sets, and for large threshold parameters NeMa suffers from excessively long query times, exceeding 1,000 seconds.

\compress
\section{Conclusion}
In this paper, we proposed a novel approach for subgraph matching based on the fused Gromov-Wasserstein distance.
Two frameworks have been presented:
First, SOT finds query graphs in a source graph by combining ideas from partial optimal transport and the fused Gromov-Wasserstein distance.
Then, SSOT extends SOT and significantly improves query times, especially in large graphs.
In our experiments, the methods demonstrate the ability to outperform state-of-the-art methods for subgraph matching.
In particular, our methods exhibit robustness with respect to noise in the datasets and achieve very fast query times.

\compress
\section*{Acknowledgment}

This work was supported by the Knut and Alice Wallenberg
foundation under grant KAW 2021.0274.
The authors would like to thank  Geert Leus, Jos H. Weber, Lihui Liu, and Daniel Staff for valuable discussions.

\bibliography{Wasserstein2,bibliography}
\bibliographystyle{ieeetr}

\appendices
\onecolumn

%\begin{figure*}
\begin{center}
{\Huge  Appendix to\\ \vspace{10pt}
 Subgraph Matching via Partial Optimal Transport
}
\vspace{25pt}

\large Wen-Xin Pan, Isabel Haasler, Pascal Frossard \\
Signal Processing Laboratory 4, EPFL\\
Emails:\{wenxin.pan, isabel.haasler, pascal.frossard\}@epfl.ch
\vspace{42pt}
\end{center}
%\end{figure*}

\begin{addmargin}[10pt]{10pt}

\fontsize{11pt}{13pt} \selectfont

%\newpage

%\begin{multicols}{2}

In the following we provide details on several aspects of our work.
More precisely, in Appendix~\ref{app:exp} we provide a more thorough analysis of our experiments on the real world data in Section~\ref{subsec:expreal}.
In Appendix~\ref{sec:alg} we describe the Frank-Wolfe algorithm adopted for our work in detail.
Finally, in Appendix~\ref{sec:implementation} we describe further implementation details.

\vspace{10pt}

\section{\large Numerical results for real-world datasets}\label{app:exp}

\vspace{5pt}

In this section, we provide some background and more detailed results on the real-world dataset experiments in Section~\ref{subsec:expreal}. 

Table~\ref{table:datasets} summarizes several graph statistics of the studied real-world datasets. We include two datasets of real-valued features and two datasets of integer-set features. 
\begin{table}[b]
    \centering
        \caption{\large Statistics of the real-world datasets studied in Section~\ref{subsec:expreal}. }
    \begin{tabular}{c|ccccccc}
    & & & & & & & includes \\
         dataset & contents & \#graphs & average size & \#edges & average node degree &node features & self-loops?  \\
         \hline
         BZR & chemical compounds & 405 & 35.75 & 39.36 & 2.20 & 3D real-valued vectors & No \\
         FIRSTMM\_DB & 3D point cloud data & 41 & 1,377.27 & 3,074.10 & 4.42 & real-valued scalars & Yes \\
         LastFM & social networks & 1 & 7,624 & 27,806 & 7.29 & integer sets & No \\
          Deezer & social networks & 1 & 28,281 & 92,752 & 6.56 & integer sets & No
    \end{tabular}
    \label{table:datasets}
\end{table}
The node features in BZR and LastFM are especially expressive, and thus all methods achieve extremely high success rates on these datasets.
Moreover, we note that the Deezer dataset contains the largest graph, and thus results in the longest query times in most experiments. 
However, the filtering step in SSOT is able to exploit the feature information effectively, still allowing for extremely fast query times in the noise-free setting, see Table~\ref{tab:mr_exact}.
The LastFM and Deezer datasets contain the densest graphs, which results in a large increase in computation time for SSOT in the setting with noise. This is because in the noisy setting the threshold parameter $T^W$ cannot be picked to be very small, and thus a large number of optimization problems have to be solved, which are relatively large due to the density of the graphs.

In the experiments on the real world datasets with noise, we adaptively chose the threshold parameter $T^W$ in SSOT and the node feature threshold parameter in NeMa.
In Figure~\ref{ratio-time} we show the success rate and average query time for several different choices of tested threshold parameters in $\{1,0.5,0.1,0.05,...,10^{-9}\}$.
As observed in Section~\ref{subsec:expreal}, SSOT achieves very good success rates and query times. In general, a higher success rate can be achieved by increasing the threshold parameter $T^W$, which results in longer query times.
In contrast, we observe the opposite effect of NeMa's performance on FIRSTMM\_DB. Here, the highest success rate is achieved for the threshold parameter that gives the fastest query time. This unintuitive behavior may make it difficult to tune the parameter for NeMa in practice.
Finally, we note that although SSOT generally achieves the highest success rates, in some cases SOT may achieve satisfying results at a much lower computational cost. For instance, this is the setting for the relatively dense dataset LastFM. Here, SOT achieves a success rate of 90\% within a few seconds of query time. At the same computational time, SSOT has a very small success rate, and the required query time to achieve a satisfying success rate is of an order 10 higher.

\begin{figure*}
	\centering
    \begin{subfigure}[b]{0.45\textwidth}
		\centering
		 \includegraphics[width=\textwidth]{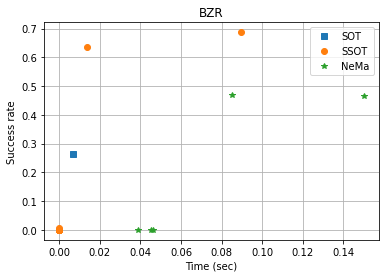}
		% \caption{BZR dataset.}
		% \label{4fig-bzr}
	\end{subfigure} %\hfill
    \begin{subfigure}[b]{0.45\textwidth}
		\centering
		\includegraphics[width=\textwidth]{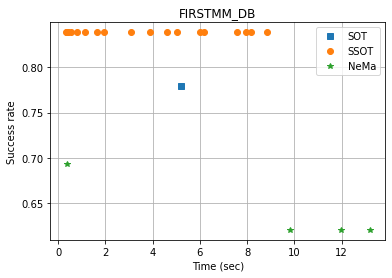}
		% \caption{FIRSTMM\_DB dataset.}
		% \label{4fig-firstmm}
	\end{subfigure}
	\begin{subfigure}[b]{0.45\textwidth}
		\centering
		\includegraphics[width=\textwidth]{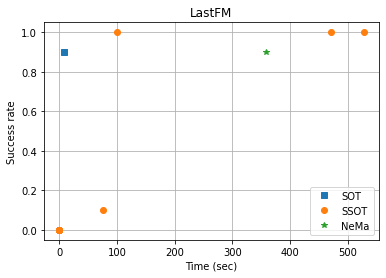}
		% \caption{LastFM dataset.}
		% \label{4fig-lastfm}
	\end{subfigure} %\hfill
    \begin{subfigure}[b]{0.45\textwidth}
		\centering
		\includegraphics[width=\textwidth]{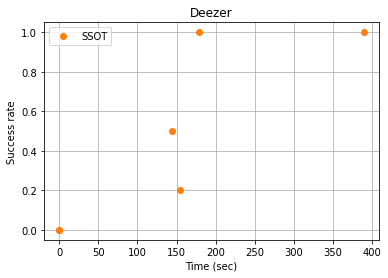}
		% \caption{Deezer dataset.}
		% \label{4fig-deezer}
	\end{subfigure}
    \caption{\large Success rates versus query times for different threshold
parameters.}
	\label{ratio-time}
\end{figure*}

\vspace{10pt}

\section{\large Frank-Wolfe algorithm}\label{sec:alg}

\vspace{5pt}

This section introduces the optimization algorithm for our proposed frameworks SOT and SSOT.
Following previous works on the fused Gromov-Wasserstein problem \cite{titouan2019optimal}, we utilize the Frank-Wolfe algorithm, also called conditional gradient method \cite{jaggi_revisiting_nodate}.
This is a first order optimization method, which takes gradient steps in the feasible set if the initialization is feasible.
We note that the gradient of the objective function
\begin{equation} \label{eq:J}
    \mathcal{J}(\hat{\boldsymbol{T}}) = (1-\alpha)\langle\hat{\boldsymbol{T}}, \hat{\boldsymbol{M}}\rangle+\alpha \langle \hat{\boldsymbol{L}} \otimes \hat{\boldsymbol{T}}, \hat{\boldsymbol{T}} \rangle 
\end{equation}
is given by
\begin{equation} \label{eq:Jgrad}
    \nabla \mathcal{J}(\hat{\boldsymbol{T}}) =(1-\alpha)\hat{\boldsymbol{M}}+\alpha \cdot 2 \cdot \left(\hat{\boldsymbol{L}} \otimes \hat{\boldsymbol{T}}\right),
\end{equation}
see \cite{titouan2019optimal} for details.
The Frank-Wolfe algorithm for the SOT method is summarized in Algorithm~\ref{alg:FW}.

\begin{figure}[tb]
  \centering
  \begin{minipage}{.9\linewidth}
\RestyleAlgo{ruled}
\SetKwComment{Comment}{/* }{ */}
\begin{algorithm}[H]
\normalsize
	\caption{\large Frank-Wolfe method for SOT}
        \label{alg:FW}
	\KwIn{$ {G}_s $, $ {G}_q $; Convergence tolerance $\delta$}
	\KwOut{$G^*$}
	$\boldsymbol{p} \leftarrow\frac{1}{n}\boldsymbol{1}_n$\;
  $ \hat{\boldsymbol{q}} \leftarrow\left[( \frac{1}{n} \boldsymbol{1}_m )^\top, 1-\frac{m}{n} \right]^\top$\;
Define ${\boldsymbol{L}}$ as in \eqref{eq:L_cost} \;
 $\hat{\boldsymbol{T}}^{(0)} \leftarrow \boldsymbol{p}\hat{\boldsymbol{q}}^\top$ \label{T0}\;
	
	\While{$ \Bigr\lvert \mathcal{J}\left(\boldsymbol{T}^{(k+1)}\right) -\mathcal{J}\left(\boldsymbol{T}^{(k)}\right) \Bigr\rvert \geq \delta$}
	{
        Compute $ \left(\hat{\boldsymbol{L}} \otimes \hat{\boldsymbol{T}}^{(k)} \right) $ as in Proposition \ref{tensor-matrix-isolation}  \label{tensor-matrix-Alg2}\;
        Compute $\nabla \mathcal{J}\left(\hat{\boldsymbol{{T}}}^{(k)} \right)$ as in \eqref{eq:Jgrad}  \;
	$\bar{\boldsymbol{T}}^{(k)} \in \underset{\hat{\boldsymbol{T}} \in \hat{\mathcal{T}}(\boldsymbol{p}, \boldsymbol{q})}{\arg \min}  \Bigr\langle \nabla \mathcal{J}\left(\hat{\boldsymbol{T}}^{(k)}\right), \hat{\boldsymbol{T}}^{(k)} \Bigr\rangle$\;
		$ \boldsymbol{d}^{(k)} \leftarrow  \bar{\boldsymbol{T}}^{(k)} - \hat{\boldsymbol{T}}^{(k)} $\;
		$ \gamma^{(k)} \leftarrow \underset{\gamma \in[0,1]}{\arg \min } \mathcal{J}\left(\hat{\boldsymbol{T}}^{(k)}+\gamma \boldsymbol{{d}}^{(k)}\right) $ \;
		$ \hat{\boldsymbol{T}}^{(k+1)} \leftarrow \hat{\boldsymbol{T}}^{(k)}+\gamma^{(k)} \boldsymbol{d}^{(k)} $\;
	}
	$ \hat{\boldsymbol{T}}^* \leftarrow \hat{\boldsymbol{T}}^{(k)} $ \;
    Return the matched subgraph $G^*$ obtained with $\hat{\boldsymbol{T}}^*$, and the objective value $ \mathcal{J}\left(\hat{\boldsymbol{{T}}}^{*} \right) $, computed as in \eqref{eq:J}.
\end{algorithm}
  \end{minipage}
\end{figure}
We note that by construction the initial transport plan  $\hat{\boldsymbol{T}}^{(0)} = \boldsymbol{p}\hat{\boldsymbol{q}}^\top$, as defined in line~\ref{T0} of the algorithm, lies in the set $\mathcal{T}(\boldsymbol{p},\hat{\boldsymbol{q}})$. Thus in each iteration of the Frank-Wolfe algorithm, we get a feasible transport plan $\hat{\boldsymbol{T}}^{(k)} \in \mathcal{T}(\boldsymbol{p},\hat{\boldsymbol{q}})$.
The optimal matching between the nodes in the source graph and the query graph can be reconstructed from the non-zero entries of the optimal transport matrix $\hat{\boldsymbol{T}}^*$.

The computational bottleneck of the algorithm is the computation of the tensor-matrix product in line~\ref{tensor-matrix-Alg2} of Algorithm~\ref{alg:FW}, as defined in \eqref{eq:tensormatrix}.
Note that this summation requires $\mathcal{O}\left(n^2 m^2\right)$ operations.
In some special cases, structures in the cost tensor $\hat{\boldsymbol{L}}$ can be exploited to perform the summation in \eqref{eq:tensormatrix} more efficiently,
reducing the complexity to $\mathcal{O}\left(n^2 m+m^2 n\right)$, see \cite[Proposition 1]{peyre_gromov-wasserstein_2016}.
By separating the dummy node we can adapt this computational trick to the partial fused Gromov-Wasserstein distance.

\begin{prop}[Adaptation of {\cite[Proposition 1]{peyre_gromov-wasserstein_2016}}]\label{tensor-matrix-isolation}
The tensor-matrix product \eqref{eq:tensormatrix} is of the form
\begin{equation}\label{tensor-matrix}
		\left(\hat{\boldsymbol{L}} \otimes \hat{\boldsymbol{T}}\right)_{i,j}= \begin{cases}
			\left({\boldsymbol{L}} \otimes {\boldsymbol{T}}\right)_{i,j},  \quad \text{for} \quad j=1,...,m \\
			0,  \quad \text{for}\quad j = m+1,
		\end{cases}
	\end{equation}
 with
\begin{equation*}\label{modified-Peyre}
        {\boldsymbol{L}} \otimes {\boldsymbol{T}}
		= (\boldsymbol{C}^s \odot \boldsymbol{C}^s)
  {\boldsymbol{T}}\mathbf{1}_m \mathbf{1}_{m}^{\top}
		+\mathbf{1}_{n} {\boldsymbol{q}}^{\top}  (\boldsymbol{C}^q \odot \boldsymbol{C}^q)^{\top}	-2 \boldsymbol{C}^s {\boldsymbol{T}} (\boldsymbol{C}^q)^{\top},
\end{equation*} 
where $\odot$ denotes elementwise multiplication.
\end{prop}

\begin{proof}
     Recall that $ \hat{\boldsymbol{L}}_{i,i',j,j'} $ defined in \eqref{eq:Lhat} is zero if $ j = m+1  $ or $ j'=m+1 $. Thus, all the entries in the last column of $ \hat{\boldsymbol{L}} \otimes \hat{\boldsymbol{T}} \in \mathbb{R}^{n\times (m+1)} $ will also be zeros. 
    More precisely, for $j=m+1$, the corresponding elements in $\hat{\boldsymbol{L}} \otimes \hat{\boldsymbol{T}}$ are given by
    \begin{equation*}
        \left(\hat{\boldsymbol{L}} \otimes \hat{\boldsymbol{T}}\right)_{i,m+1 }
    	=\sum_{i',j'} \hat{\boldsymbol{L}}_{i,i', m+1,j'} \hat{\boldsymbol{T}}_{i', j'}
    	=0.
    \end{equation*}
    Moreover, for $j\le m$ we get 
    \begin{align*}
    	\begin{aligned}
    		\left(\hat{\boldsymbol{L}} \otimes \hat{\boldsymbol{T}}\right)_{i,j}
    		&= \sum_{i'=1}^n \sum_{j'= 1}^{m+1} \hat{\boldsymbol{L}}_{i,i',j,j'} \hat{\boldsymbol{T}}_{i', j'}
       = \sum_{i'=1}^n \left( \sum_{j'= 1}^{m}  \hat{\boldsymbol{L}}_{i,i',j,j'} \hat{\boldsymbol{T}}_{i', j'} + \hat{\boldsymbol{L}}_{i,i',j,m+1} \hat{\boldsymbol{T}}_{i', m+1} \right) 
       	= \sum_{i'=1}^n \sum_{j'= 1}^{m} \boldsymbol{L}_{i,i',j,j'} \boldsymbol{T}_{i', j'}\\
    	&	= \left({\boldsymbol{L}} \otimes {\boldsymbol{T}}\right)_{i,j}. \\
    	\end{aligned}
    \end{align*}
    Finally, since $\boldsymbol{L}$, as defined in \eqref{eq:L_cost}, can be written as
    \begin{equation*}
    \begin{aligned}
    \boldsymbol{L}^{G_s,G_t}_{i,i',j,j'}  = \left( \boldsymbol{C}^s_{i,i'} - \boldsymbol{C}^t_{j,j'} \right)^2  
    = (\boldsymbol{C}^s_{i,i'})^2 + (\boldsymbol{C}^t_{j,j'})^2 - 2\boldsymbol{C}^s_{i,i'}\boldsymbol{C}^t_{j,j'},
    \end{aligned}
    \end{equation*}
    we can apply \cite[Proposition 1]{peyre_gromov-wasserstein_2016} and get the expression for the tensor-matrix product $ {\boldsymbol{L}} \otimes {\boldsymbol{T}}$ in the proposition.
\end{proof}

\vspace{10pt}

 \section{\large Implementation details} \label{sec:implementation}

\vspace{5pt}

 In this section, we provide details regarding our numerical implementation.
More precisely, we describe the normalized fused Gromov-Wasserstein problem, which allows us to control the importance of features versus structure in a practical way.
Finally, we provide details on the choice of parameters used in the experiments.

 \subsection{\large Normalized fused Gromov-Wasserstein problem}
 \label{app:nFGW}

Note that the partial fused Gromov-Wasserstein distance is zero if an exact match is found. However, in the presence of noise, it is crucial to tune the importance of features and structure in a controlled way. 
In order to emphasize features and structure equally when setting the trade-off parameter to $\alpha =0.5$, we normalize both terms of the objective in \eqref{eq:partialFGW}.
Recall that we use the structure cost defined in \eqref{eq:L_cost}, where the structure matrices $\boldsymbol{C}^s$ and $\boldsymbol{C}^t$ are defined as the graphs' adjacency matrices. The elements in the tensor $\boldsymbol{L}$ thus lie in the set $\{0,1\}$.
Moreover, for the feature costs defined in $\boldsymbol{M}$ we adopt normalized distances as discussed in Section~\ref{sec:exp}.
The elements of this matrix thus lie within $[0,1]$. 
We thus define the normalized fused Gromov-Wasserstein problem for subgraph matching, as the partial fused Gromov-Wasserstein problem
\begin{equation} \label{eq:npfgw}
    \underset{{\boldsymbol{T}}\in \mathcal{T}_{m/n}(\boldsymbol{p},{\boldsymbol{q}})}{\min } \ (1-\alpha) \frac{n}{m} \langle {\boldsymbol{T}}, {\boldsymbol{M}}\rangle  + \alpha \frac{n^2}{{m^2}} \langle {\boldsymbol{L}} \otimes {\boldsymbol{T}}, {\boldsymbol{T}} \rangle .
\end{equation}
Following the same strategy as in Section~\ref{sec:sot} we can augment the target distribution by a dummy node and formulate problem \eqref{eq:npfgw} as a standard fused Gromov-Wasserstein problem \eqref{eq:FGW}, where the first term in \eqref{eq:FGW} is multiplied by a factor $n/m$ and the second term is multiplied by a factor $n^2/m^2$.

\begin{prop}\label{prop-nFGW}
For every feasible transport plan ${\boldsymbol{T}}\in \mathcal{T}_{m/n}(\boldsymbol{p},{\boldsymbol{q}})$, each of the two terms in the normalized fused Gromov-Wasserstein problem for subgraph matching in \eqref{eq:npfgw}
	takes a value in $ [0,1] $. 
\end{prop}
\begin{proof}[Proof of Proposition \ref{prop-nFGW}]
First note that the two terms in \eqref{eq:npfgw} are always non-negative. Thus, we only need to show the upper bound.
Note that since we use normalized feature costs, it holds that ${\boldsymbol{M}}_{i,j}\le 1$, and thus
	\begin{equation*}\label{key}
		\sum_{\substack{{i=1,...,n} \\ {j=1,...,m}}} {\boldsymbol{M}}_{i,j}{\boldsymbol{T}}_{i,j}
		\le \sum_{\substack{{i=1,...,n} \\ {j=1,...,m}}}
		{\boldsymbol{T}}_{i,j}= \sum_{j=1,...,m}{\boldsymbol{q}}_j\\
  = \frac{m}{n}.
	\end{equation*}
Similarly, since the elements in $\boldsymbol{L}$ lie in the set $\{0,1\}$, we get that
\begin{equation*}
   \sum_{\substack{{i,i'=1,...,n} \\ {j,j'=1,...,m}}}{\boldsymbol{L}}_{i,i',j,j'} {\boldsymbol{T}}_{i,j} {\boldsymbol{T}}_{i',j'}
		\le \sum_{\substack{{i,i'=1,...,n} \\ {j,j'=1,...,m}}}{\boldsymbol{T}}_{i,j} {\boldsymbol{T}}_{i',j'} 
  =\sum_{\substack{{i=1,...,n} \\ {j=1,...,m}}} \left({\boldsymbol{T}}_{i,j}
		\sum_{\substack{{i'=1,...,n} \\ {j'=1,...,m}}} {\boldsymbol{T}}_{i',j'} \right)
		= \frac{m}{n}\sum_{\substack{{i=1,...,n} \\ {j=1,...,m}}}{\boldsymbol{T}}_{i,j}
		= \frac{m^2}{n^2}.
\end{equation*}
Thus, the result follows.
\end{proof}

\subsection{\large Further implementation details}

In the experiments, we set the convergence tolerance in Algorithm~\ref{alg:FW} to
$ \delta = 10^{-9}$. 
We note that this is a very conservative choice for the parameter.
We may achieve faster query times for our proposed methods, while still maintaining good performances, by increasing this tolerance.

The implementations of NeMa and G-Finder were slightly modified to allow for all cases tested in our experiments. The NeMa implementation named \textsf{fornax} does not support the case where the source graph contains self-loops. In this case, NeMa's success rate is marked as zero. G-Finder does not support the case where the query graph is a line graph. We manually add an additional function to support this case.

\end{addmargin}

\end{document}